\documentclass[letter,reqno]{clt}

\title{Law of large numbers and central limit theorem for ergodic quantum processes}
\author{Lubashan Pathirana}
\email{pathiran@msu.edu}
\author{Jeffrey Schenker}
\email{schenke6@msu.edu}
\address{Department of Mathematics, Michigan State University, 619 Red Cedar Road, East Lansing MI, 48823}

\begin{document}

\pagenumbering{arabic}
\lhead{\thepage}
\maketitle
\begin{abstract}
    A discrete quantum process is represented by a sequence of quantum operations, which are completely positive maps that are not necessarily trace preserving. We consider quantum processes that are obtained by repeated iterations of a quantum operation with noise. Such ergodic quantum processes generalize independent quantum processes. An ergodic theorem describing convergence to equilibrium for a general class of such processes was recently obtained by Movassagh and Schenker in \cite{movassagh2022ergodic, movassagh2021theory}. Under irreducibility and mixing conditions we obtain a central limit type theorem describing fluctuations around the ergodic limit.
\end{abstract}


\section{Introduction and main results}
A \emph{quantum channel} (QC) is a linear, completely positive, and trace preserving map on the trace class operators, where the state of the system is represented by a non-negative operator of trace one \textemdash \ a \emph{density matrix}.  Such maps can describe the evolution of an open quantum system over a discrete unit of time, including averaged effects of measurements and environmental noise.  More generally, one introduces \emph{quantum operations}(QOs) \textemdash \ completely positive and trace non-increasing maps \textemdash \ to describe processes with loss or which happen only with a certain probability. A \emph{quantum process} is a sequence of QOs describing the evolution of the system over a consecutive sequence of time intervals.  Quantum processes represent the most general description of the average evolution of an open quantum system neglecting memory effects in the environment.

In a pair of recent papers \cite{movassagh2022ergodic,movassagh2021theory}, Movassagh and the second author formulated the notion of an \emph{ergodic quantum process} in which the individual QOs are obtained by sampling a QO valued function along a trajectory of an ergodic dynamical system.  For processes on a   finite dimensional Hilbert space and satisfying a physically natural decoherence condition, they proved convergence of the density matrix to a stationary, ergodic sequence of density matrices as time goes to infinity.  This theorem of \cite{movassagh2021theory} generalizes a result of Hennion \cite{hennion1997limit} on products of non-negative random matrices and is closely related Oseledec's multiplicative ergodic theorem \cite{oseledec1968multiplicative}. 

The results of \cite{movassagh2021theory} require essentially only decoherence and ergodicity.  In the present paper, we examine processes that satisfy stronger integrability and mixing conditions. We prove a law of large numbers and a central limit theorem for the expectation values of observables in states evolving under such a processes.  Although our main interest is in the application of these results to quantum processes, the results themselves do not require the maps to be trace non-increasing and require only \emph{positivity} (not complete positivity).  

This paper is organized as follows:
\begin{enumerate}
    \item In \S\ref{sec:intro} we state our main results after formulating certain background notions.
    \item In \S\ref{sec:background}, review some definitions and arguments from \cite{movassagh2022ergodic} that are fundamental to the proofs of our main results.
    \item In \S\ref{sec:lln}, we prove Theorem \ref{lln}.
    \item In \S\ref{sec:proofCLT}, we prove Theorem \ref{CLT}.
    \item In \S\ref{sec:mixing}, we prove Theorem \ref{customthm3}, which gives sufficient conditions for the main hypothesis of Theorem \ref{CLT} to hold.
\end{enumerate}

\subsection*{Acknowledgements} This material is based upon work supported by the National Science Foundation under Grants No. 1900015 and 2153946.

\section{Formal statement of the main results}\label{sec:intro}

\subsection{Positive Linear Maps}
 Let $\mathbb{M}_D = \mathbb{C}^{D\times D}$ denote the space of $D\times D$ matrices.    We consider the space $\mathbb{M}_D$ with its standard topology as a finite-dimensional vector space.  For definiteness, we take this to be the norm topology generated by the \emph{trace norm}, $\norm{A} := \text{Tr}\sqrt{(A^*A)}$ for any $A \in \mathbb{M}_D$, but of course the topology is independent of the norm (since $\mathbb{M}_D$ is finite dimensional). For any matrix $A \in \mathbb{M}_D$ we denote by $A^*$ the adjoint matrix (conjugate transpose).
 
 The space of linear operators on $\mathbb{M}^D$ will be denoted by $\mapspace$. We equip the space $\mapspace$ with the operator norm induced by the trace norm on $\mathbb{M}_D$. That is, for $\phi \in \mapspace$:
     \begin{equation}\label{normForSuperOperators}
        \norm{\phi} = \sup \{\norm{\phi(A)} : A \in \mathbb{M}_D, ||A|| = 1\}	\ .
     \end{equation}
For any $\phi \in \mapspace$ the adjoint of $\phi$ is the unique map $\phi^* \in \mapspace$ determined by the identity:
     \begin{equation}
        \inner{A}{\phi(B)} = \inner{\phi^*(A)}{B}\text{ for all } A,B \in \mathbb{M}_D \ ,
    \end{equation}
where $\inner{A}{B}$ denotes the Hilbert-Schmidt inner product,
\begin{equation}
    \inner{A}{B} = \tr A^* B \ .
\end{equation}
	 
 We recall that a map $\phi\in \mapspace$ is \emph{positive}, if it maps the set of positive semi-definite matrices to itself.  It is convenient to introduce notation for certain subsets of positive semi-definite matrices as follows:
 \begin{enumerate}
    \item $\pos$ is the set of all positive semi-definite $D\times D$ matrices,
    \item $\pos^0$ is the set of all positive definite $D\times D$ matrices,
    \item $\mathbb{S}_D$ is the set of positive semi-definite $D\times D$ matrices with trace one, and
    \item $\mathbb{S}_D^o$ is the set of positive definite $D\times D$ matrices with trace one. 
 \end{enumerate}
 The subset $\mathbb{S}_D$, being bounded and closed, is compact by the Heine-Borel theorem. Note that $\phi$ is positive if and only if $\phi(\mathbb{S}_D)\subset \pos$. We call $\phi$ \emph{strictly positive} if $\phi(\mathbb{S}_D)\subset \pos^0$.
	 	 
 Positive maps satisfy a generalization of the Perron-Frobenius Theorem (see \cite{kreinLinearOperatorsLeaving1950,evans1978spectral}): every such map $\phi$ has an eigenmatrix $R\in \mathbb{S}_D$ with eigenvalue equal to the spectral-radius $r(\phi)$.   The map $\phi$ is called \emph{irreducible} if $(\mb{1}+\phi)^{n}$ is strictly positive for some $n$.\footnote{Equivalently, no \emph{hereditary sub-algebra},  $P\bb{M}_DP$ with $P$ an orthogonal projection, is invariant under $\phi$.  See \cite{evans1978spectral}.} By \cite[Theorems 2.3 \& 2.4]{evans1978spectral} we have the following
 \begin{prop} If $\phi$ is an irreducible positive map, then there is a unique $R \in \mathbb{S}_D$ such that $\phi(R)= \Lambda R$ for some $\Lambda \in \mathbb{C}$.  Furthermore, the eigen-matrix $R$ is non-singular ($R\in \mathbb{S}_D^\circ$) and the eigenvalue $\Lambda=r(\phi) >0$ is the spectral radius of $\phi$.
 \end{prop}
We call the unique eigenmatrix $R\in \mathbb{S}_D$ of an irreducible map $\phi$ the \emph{right Perron-Frobenius eigenmatrix of $\phi$}.  The map $\phi$ also has a \emph{left Perron-Frobenius eigenmatrix}, which is the Perron-Frobenius eigenmatrix of $\phi^*$.  (Note that $\phi$ is irreducible if and only if $\phi^*$ is.)  

The Perron-Frobenius eigenmatrix $R$ of an irreducible map $\phi$ may be interpreted as a fixed point of the \emph{projective action of $\phi$ }:
\begin{equation}
    \phi\cdot X = \dfrac{\phi(X)}{\tr \phi(X)}  \ . 
\end{equation}
For a general map, the projective action is defined for $X\in\mathbb{S}_D \setminus \ker \phi$.  However, if $\ker \phi \cap \mathbb{S}_D=\emptyset$ then the projective action is defined on all of $\mathbb{S}_D$.  As this condition will play a key role in our analysis, we make the following
\begin{definition} 
    A positive linear map $\phi\in \mathcal{L}(\mathbb{M}_D)$ is \emph{non-destructive} if $\ker \phi \cap \mathbb{S}_D=\emptyset$.  If $\phi^*$ is non-destructive, we say that $\phi$ is \emph{non-transient}.
\end{definition}

The terminology \emph{non-transient} stems from the fact that if $\rho\in \ker\phi^*\cap \mathbb{S}_D$ and $P$ is the projection onto $\ran \rho$, then $\phi^*(P)=0$ and $\ran \phi$ is contained in the hereditary sub-algebra $P^\perp \mathbb{M}_DP^\perp$ where $P^\perp=I-P$.  Thus the subspace corresponding to $\ran P$ is a ``transient subspace'' for $\phi$. 

A sufficient condition for $\phi$ to be non-destructive \emph{and} non-transient is that $\phi^n$ be strictly positive for some $n>0$.  This condition is, in turn, equivalent to $\phi$ being \emph{irreducible and aperiodic}, i.e., irreducible and having no eigenvalues on the circle $\{|z|=r(\phi)\}$ except for the Perron-Frobenius eigenvalue.

\subsection{Limiting results for eigenmatrices of ergodic quantum processes}
As in \cite{movassagh2022ergodic}, we are interested in sequences $\Phi^{(n)}$ such that
\begin{equation}
     \Phi^{(n)} = \phi_n \circ\ldots \phi_1 \quad \text{ with } \quad \phi_n=\phi_{0;\theta^n\omega} \ ,
\end{equation}
where $\omega \mapsto \phi_{0;\omega}$ is a positive map valued random variable defined on a probability space $(\Omega,\mathcal{F},\prob)$ and $\theta:\Omega \to \Omega$ is an ergodic map. We recall that a measurable map $\theta:\Omega\to \Omega$  is
\begin{enumerate} 
    \item \textit{measure preserving}  if $\prob(\theta^{-1}(A)) = \prob(A)$ for all $A \in \mathcal{F}$, and
    \item \emph{ergodic} if it  is measure preserving and $\prob(A)=0$ or $1$ whenever $\theta^{-1}(A) = A$.
\end{enumerate}
We further recall that either of the following two conditions is sufficient for a measure preserving map $\theta$ to be ergodic:
\begin{enumerate}
    \item \emph{essentially $\theta$-invariant sets have measure $0$ or $1$}, i.e., $\prob(A)=0$ or $1$ whenever $A\in \mathcal {F}$ with $\prob(A\Delta\theta^{-1}(A))=0$.
    \item \emph{essentially $\theta$-invariant functions are almost surely constant}, i.e., if $f\circ \theta =f$ almost surely, then there is $c\in \mathbb{R}$ such that $f=c$ almost surely.    
\end{enumerate}
See \cite{walters2000introduction} for proofs of these facts and further discussion of ergodic maps.	
	
Now fix a probability space $(\Omega,\mathcal{F},\prob)$ and an ergodic map $\theta:\Omega\to \Omega$.  For a random variable $X:\Omega\to \mathcal{S}$, with $\mathcal{S}$ some measurable space, we denote the value of $X$ at $\omega\in \Omega$ by $X_\omega$, and will often omit $\omega$ from the notation for simplicity. This subscript notation is convenient as we consider map valued random variables which take a matrix as an argument. Let $\varphi_0 :\Omega \to \mapspace$ be a positive map valued random variable, where we take the Borel $\sigma$-algebra on $\mapspace$. For each $n \in \mathbb{N}$, define $\varphi_{n;\omega}  = \varphi_{0;\theta^n(\omega)}$.  
Let 
\begin{equation}\label{Phi}
    \Phi^{(n)}_\omega = \varphi_{n;\omega} \circ \varphi_{n-1;\omega} \circ \cdots \varphi_{1;\omega}.
\end{equation}
For $k \geq 0$, we have 
\begin{equation}
    \Phi_{\theta^k(\omega)}^{(n)}  = \varphi_{n;\theta^k(\omega)} \circ \ldots \varphi_{1,\theta^k(\omega)} = \varphi_{n+k;\omega}\circ \cdots \varphi_{1+k;\omega} \ ;
\end{equation}
as above we may omit $\omega$ from the notation and simply write this as $\Phi_{\theta^k}^{(n)} = \varphi_{n+k}\circ \cdots \varphi_{1+k} $. 

In the present work, we study sequences $\Phi^{(n)}$ with the property that $\Phi^{(n)}$ is eventually strictly positive.  We denote by $\tau_\omega$ the time at which $\Phi^{(n)}_\omega$ becomes strictly positive  and stays strictly positive thereafter:
\begin{equation}\label{tau}
    \tau_\omega  = \inf \{n \geq 1 : \Phi_\omega ^{(n+k)} \text{ is strictly positive } \forall k \geq 0\} \ .
\end{equation} 
Our first assumption is that $\tau<\infty$ almost surely:
\begin{assumption}{1}\label{1} 
    We have $\prob\{\tau < \infty\} = 1$, i.e.,  the sequence $\Phi^{(n)}$ is \emph{almost surely eventually strictly positive}.
\end{assumption}

Assumption \ref{1} was also the main assumption of \cite{movassagh2022ergodic}, where it was shown to be equivalent to the following two conditions provided that  $\theta$ is invertible (see \cite[Lemma 2.1]{movassagh2022ergodic}):
\begin{enumerate}
    \item there exists $N_0\in\mathbb{N}$ such that $\prob(\Phi^{(N_0)} \text{ is strictly positive)} >0$, and
    \item $\prob \{ \varphi_0 \text{ is non-destructive and non-transient } \} = 1$. 
\end{enumerate}
One consequence of this equivalence is that, if $\theta$ is invertible and Assumption \ref{1} holds, then $\tau$ can be expressed as
\begin{equation} \label{eq:vartau}
    \tau = \inf \{n \ge 1 : \Phi^{(n)} \text{ is strictly positive } \} \ .
\end{equation}
In particular, $\tau$ is then a \emph{stopping time} with respect to the filtration $(\mc{F}_n)_{n=0}^\infty$ where $\mc{F}_n$ denotes the $\sigma$-algebra generated by $\phi_0,\ldots,\phi_n$.

Since any strictly positive map is irreducible, Assumption \ref{1} guarantees that the left and right Perron-Frobenius eigenmatrices, $R_n$ and $L_n$, exist for for sufficiently large $n$:
    \begin{equation}\label{LAndR}
        \Phi^{(n)}(R_n) = \Lambda_n R_n \quad \text{ and } \quad \Phi^{(n)*}(L_n) = \Lambda_n L_n \ .
    \end{equation}
Here $\Lambda_n=\Lambda_{n;\omega}$ denotes the spectral radius of $\Phi^{(n)}$ and $L_n$, $R_n$ are $\mathbb{S}_D^\circ$ valued random variables, i.e., they are $D\times D$ positive definite matrix valued random variables with $\tr R_n=\tr L_n=1$. We have the following 
	
\begin{lemma}[{\cite[Theorem 1]{movassagh2022ergodic}}]\label{existanceofz}
Let $(\varphi_n)_{n\geq1}$ and $\Phi^{(n)}$ be as in eq.\ \eqref{Phi} and let $L_n$ be as in eq.\ \eqref{LAndR}. If Assumption \ref{1} holds, then there is an $\mathbb{S}_D^\circ$ valued random variable $Z_1'$ such that
\begin{equation} 
    Z_1 \ \overset{a.s.}= \ \lim_{n\to \infty} L_n \ 
\end{equation}
and, with $Z_k:= Z_1\circ \theta^{k-1}$, we have  for every $k\in \mathbb{N}$, $\varphi_k^*\cdot Z_{k+1}=Z_k$ a.s..
\end{lemma}
\begin{remark} This is half of \cite[Theorem 1]{movassagh2022ergodic}.  The other half involves the convergence of the right eigenvectors and requires invertibility of the ergodic map $\theta$.  A close reading of the proof (see \cite[Lemma 3.12]{movassagh2022ergodic}) shows that invertibility of $\theta$ is not necessary for the portion stated here.
\end{remark}

 \subsection{Law of Large Numbers}
Our first main result is concerned with expectations of the form $\inner{Y}{\Phi^{(n)}(X)}$ with $X,Y\in \mathbb{S}_D$.  The main idea here is that for large $n$,  the Perron-Frobenius eigenvalue $\Lambda_n$ of $\Phi^{(n)}$ typically exhibits exponential growth or decay and dominates the expression, so that we expect 
\begin{equation} \inner{Y}{\Phi_n(X)} \ \approx \ \Lambda_n \frac{\inner{Y}{R_n}\inner{L_n}{X}}{\inner{L_n}{R_n}} \ + \ \text{lower order terms},\label{eq:heuristic1} \end{equation}
where $L_n$ and $R_n$ are the left and right Perron-Frobenius eigenmatrices, respectively, normalized so that $\tr L_n=\tr R_n=1$.  
Under Assumption \ref{1}, $L_n$ and $R_n$ are positive definite, so $\inner{Y}{R_n}\inner{L_n}{X}\neq 0$ and eq.\ \eqref{eq:heuristic1} suggests that
$$ \ln \inner{Y}{\Phi^{(n)}(X)} \ \approx  \ \ln\Lambda_n  \ + \ O(1). $$
Thus we expect a Law of Large Numbers, $\frac{1}{n}\ln \inner{Y}{\Phi^{(n)}(X)} \to l$, where $l=\lim_n\frac{1}{n}\ln \Lambda_n$.

To obtain this Law of Large Numbers, we require an integrability assumption for $\ln \|\varphi_0^*\|$ and for $\ln v(\varphi_0^*)$, where for $\phi \in \mapspace$ we define
\begin{equation}\label{eq:v}
    v(\phi) \ := \ \inf \{ \| \phi(X)\| : X \in \mathbb{S}_D\} \ . 
\end{equation}
\begin{assumption}{2} \label{2}
    We have $\avg[|  \ln \|\varphi_0^*\|  |] < \infty$ and $\avg[|\ln v(\varphi_0^*)  |] < \infty$\ . 
\end{assumption}
\begin{remark} 
    We note that any non-destructive map $\phi$ (in particular, any strictly positive map) must have $v(\phi) > 0$ because $\mathbb{S}_D$ is a compact set and the map $A \mapsto \norm{\phi(A)}$ is continuous.
\end{remark}

With Assumptions \ref{1} and \ref{2} we have the following
\begin{customthm}{1 - Law of Large Numbers}\label{lln}
    Let $\Phi^{(n)}$ be a random sequence of positive maps as in eq.\ \eqref{Phi}.  If Assumptions \ref{1} and \ref{2} hold then
        \begin{equation}\label{eq:firstpart}
            \lim_{n \to \infty} \ \sup_{X,Y \in \mathbb{S}_D} \Big|\frac{1}{n}\ln \inner{Y}{\Phi^{(n)}(X)} - l\Big| \ = \ 0 \quad \text{a.s.} ,
        \end{equation}	
    where $l = \avg[\ln \|\varphi_0^*(Z_1)\|]$ with $Z_1 =\lim_n L_n$. Furthermore
        \begin{equation}\label{eq:secondpart}
            \lim_{n \to \infty} \frac{1}{n}\ln \norm{\Phi^{(n)}} \ = \ \lim_{n\to\infty} \frac{1}{n} \ln \Lambda_n = l  \quad  \text{a.s.} \ , 
        \end{equation}
    with $\Lambda_n$ the Perron-Frobenius eigenvalue of $\Phi^{(n)}$. 
\end{customthm}
\begin{remark} 
    We take $\ln \inner{Y}{\Phi^{(n)}(X)}=-\infty$ if $\inner{Y}{\Phi^{(n)}(X)}=0$; by Assumption \ref{1} this happens for at most finitely many $n$.
    By Assumption \ref{2}, $l = \avg[\ln \|\varphi_0^*(Z_1)\|] $ is finite.
\end{remark} 

Theorem \ref{lln} is closely related in spirit to the Furstenberg-Kesten theorem \cite{furstenberg1960} and Oseledet's Theorem \cite{oseledec1968multiplicative} (see also \cite{goldsheid1989}). By the Furstenberg-Kesten Theorem, the following limit exsists
$$\lim_{n\to \infty} \frac{1}{n} \ln  \norm{\Phi^{(n)}} \quad \text{a.s.} \ = \ \lambda \quad{a.s.} , $$
where $\lambda$ is a deterministic quantity called the \emph{top Lyapunov exponent}  of the cocycle  $(X,n) \mapsto \Phi^{(n)}(X)$.  By Oseledet's Theorem, there is a (random) proper subspace $L \subset \mathbb{M}_D $ such that for $X\in \mathbb{M}_D \setminus L_{j+1}$ we have
$$\lim_{n\to \infty} \frac{1}{n} \ln \norm{\Phi^{(n)}(X)} \ = \ \lambda \ . $$

The identity eq.\ \eqref{eq:firstpart} is the key result in Theorem \ref{lln}.  Indeed, since $\Lambda_n = \inner{L_n}{\Phi^{(n)}(\mathbb{I})}$ it follows directly from eq.\ \eqref{eq:firstpart} that $l= \lim_n \frac{1}{n}\ln \Lambda_n$ almost surely. Furthermore, as the proof of eq.\ \eqref{eq:firstpart} will make clear (see Lemma \ref{En&Dn}), we also have $\lim_n \tfrac{1}{n}\ln \norm{\Phi^{(n)*}(Y)} \ = \ l$ a.s. for any $Y\in \mathbb{S}_D$.  Since $\operatorname{span} \mathbb{S}_D =\mathbb{M}_D$, it follows from Oseledet's Theorem that $l=\lambda$, the top Lyapunov exponent, and thus that $l= \lim_n \frac{1}{n} \ln \norm{\Phi^{(n)}}$.  Therefore eq.\ \eqref{eq:secondpart} is a consequence of eq.\ \eqref{eq:firstpart}. Thus to prove Theorem \ref{lln} it suffices to prove eq.\ \eqref{eq:firstpart}.  This is accomplished in \S\ref{sec:lln} below.

\subsection{Central Limit Theorem}\label{sec:CLT}
Our second main result is a central limit theorem for the fluctuations of $\ln \inner{Y}{\Phi^{(n)}(X)}$ around its asymptotic value $nl$.  For this result we require additional integrability for $\ln \|\varphi_{0;\omega}^*\|$ and $\ln v(\phi_{0;\omega}^*)$:
\begin{assumption}{2$_p$} \label{2'}
    For $p> 1$, the random variables $ \ln \|\varphi_{0;\omega}^*\|$ and $\ln v(\varphi_{0;\omega}^*)$ are in $L^p$ . 
\end{assumption}
		
To obtain a central limit theorem, we require the ergodic map $\theta$ to be invertible, and extend the definition of $\varphi_k$ to $k< 0$ by $\varphi_{k;\omega}=\varphi_{0;\theta^k\omega}$, just as for $k\ge 0$. Similarly we define $Z_{k;\omega}=Z_{1;\theta^{k-1}\omega}$ for $k\le 0$.   
The key quantities that describe the fluctuations are the deviations of $\ln \norm{\varphi_k^*(Z_{k+1})}$ from its mean:
\begin{equation}\label{eq:xi}
    \xi_k  \ := \ \ln \norm{\varphi_k^*(Z_{k+1})} - l \ \ , 
\end{equation}
where $l$ is as in Theorem \ref{lln}.  We also introduce the following reverse filtration $(\mathcal{F}^n)_{n\in \mathbb{Z}}$ on the probability space:  
\begin{equation}\label{eq:F^n} 
    \mathcal{F}^n \ :=  \ \text{ sigma algebra generated by } (\varphi_k)_{k\geq n} \ .
\end{equation} 

With these preliminaries, we have the following
 \begin{customthm}{2 - Central Limit Theorem}\label{CLT}
    Let $\Phi^{(n)}$ be a random sequence of positive maps as in eq.\ \eqref{Phi}. Suppose that the ergodic map $\theta$ is invertible, that Assumption \ref{1} holds, and that Assumption \ref{2'} holds for some $p\geq 2$.  
    If 
    \begin{equation}\label{qNormSummable}
        \sum_{n=1}^\infty \norm{\cond{\xi_0}{n}}_q < \infty	
    \end{equation}
    with $1/p + 1/q =1$, then for any sequences $(X_n)_{n\geq1}$ and $(Y_n)_{n\geq1}$ in $\mathbb{S}_n$, the random sequence
    \begin{equation}\label{SequenceInCLT}
        \bigg( \frac{1}{\sqrt{n}}\big( \ln \inner{Y_n}{\Phi^{(n)}(X_n)} - nl\big)\bigg)_{n\geq 1}
    \end{equation}
    converges in distribution to a centered normal random variable with variance
    \begin{equation}\label{sigma}
         \sigma^2 \ := \  \avg\bigg[ \bigg(\sum_{k\geq0} \big( \cond{\xi_{-k}}{0} -\cond{\xi_{-k}}{1}\big)\bigg)^2\bigg] \ \ge \  0 \ .
    \end{equation}
\end{customthm}
\begin{remark}
    The proof will show that $\sigma < \infty$, but we have allowed the possibility that $\sigma =0$.  If $\sigma=0$, the sequence in \ref{SequenceInCLT} converges to $0$ in distribution (and hence in probability). Else, the sequence in \ref{SequenceInCLT} converges to a centered normal law with variance $\sigma^2 >0$.
\end{remark}

We prove Theorem \ref{CLT} in \S\ref{sec:proofCLT} below.

The hypothesis eq.\ \eqref{qNormSummable} of Theorem \ref{CLT} may not be easy to verify directly. We close this section by introducing \emph{mixing conditions} that are sufficient for eq.\ \eqref{qNormSummable} to hold.  Let
\begin{equation}
    \mathcal{F}_n \ := \  \text{ sigma algebra generated by } (\varphi_k)_{k\leq n} \ .
\end{equation}
Note that $(\mathcal{F}_n)_{n\in \mathbb{Z}}$ is a filtration, i.e., $\mathcal{F}_n \subset \mathcal{F}_{n+1}$, while $(\mathcal{F}^n)_{n\in \mathbb{Z}}$ (defined above in eq.\ \eqref{eq:F^n}) is a reverse filtration, i.e., $\mathcal{F}^n \supset \mathcal{F}^{n+1}$.  We introduce the following \emph{mixing coefficients}:		
\begin{align}
    \alpha_n \ &:= \ \sup_{k\geq 0} \ \sup \bigg\{ \big| \prob(A\cap B) - \prob(A)\prob(B) \big| \ : \ A\in\mathcal{F}_k, \  B\in\mathcal{F}^{n+k} \bigg\} \label{alphan}\\
    \rho_n \ &:= \  \sup_{k\geq 0} \  \sup \bigg\{ \bigg| \frac{\avg[(Y - \avg[Y])(X - \avg[X])]}{\sigma(Y)\sigma(X)} \bigg| \ : \ Y\in L^2(\mathcal{F}_k), \  X\in L^2(\mathcal{F}^{n+k}), \ X, Y\neq0 \bigg\} \label{rhon}
\end{align}
We have the following:
\begin{customthm}{3}\label{customthm3}
    If Assumption \ref{2'} holds with $p>2$ and $\sum_{n\geq 1}\alpha_n^{(p-2)/p} < \infty$, then 	
    $$\sum_{n=1}^\infty \norm{\cond{\xi_0}{n}}_{q} < \infty \ , $$
    with $q$ the conjugate exponent to $p$.  If Assumption \ref{2'} holds with $p=2$ and $\sum_{n\geq 1}\rho_n < \infty$, then 
    $$\sum_{n=1}^\infty \norm{\cond{\xi_0}{n}}_2 < \infty$$
\end{customthm}

Theorem \ref{customthm3} is proved in \S\ref{sec:mixing} below.

\section{Background results: geometry of \texorpdfstring{$\mathbb{S}_D$}{lg}, contraction for positive maps, and ergodic arguments}\label{sec:background}
In this section we review some definitions and arguments from \cite{movassagh2022ergodic} that are fundamental to the proofs below.

\subsection{A metric on \texorpdfstring{$\mathbb{S}_D$}{lg}}
Following \cite{movassagh2022ergodic}, we define the following metric on $\mathbb{S}_D$:
\begin{equation}\label{distance}
				\dis{A}{B} \ := \ \frac{1 - m(A,B)m(B,A)}{1+m(A,B)m(B,A)} \ ,
			\end{equation}
where
\begin{equation}\label{eq:m(A,B)}
    m(A,B) \ = \ \sup \{\lambda : \lambda B\leq A\}]
\end{equation}
for $A,B\in \mathbb{S}_D$. The following lemma lists key properties of this metric (see \cite[Lemma 3.3, 3.8, 3.9 ]{movassagh2022ergodic} for further details and proofs):
\begin{lemma}\label{equvilanceOfMatrics} The function $\mathrm{d}$ defined in eq.\ \eqref{distance} is a metric on $\mathbb{S}_D$ satisfying:
\begin{enumerate}
    \item $\frac{1}{2}\norm{A-B}\le \dis{A}{B} \le 1$ for $A,B\in \mathbb{S}_D$.
    \item $\dis{A}{B}<1$ for $A,B \in \mathbb{S}_D^\circ$.
    \item If $A\in \mathbb{S}_d^o$, then $\dis{A}{B}=1$ if and only if $B\in \mathbb{S}_D\setminus \mathbb{S}_D^\circ$.
    \item The set $\mathbb{S}_D^\circ$ is open in the metric topology generated by $\text{d}$ and $(\mathbb{S}_D^\circ,\text{d})$ is homeomorphic to $\mathbb{S}_D^\circ$ in the standard topology (generated by $\mathrm{d}_1(A,B)=\norm{A-B})$.
\end{enumerate}
\end{lemma}

In the proofs below, the following simple consequence of the lower bound $\frac{1}{2}\norm{A-B}\le \dis{A}{B}$ will be useful.
\begin{lemma}\label{meanValueIneqaulity}
    Let $\phi\in\mapspace$ be a positive map with the property that $\ker{\phi}\cap\mathbb{S}_D = \emptyset$. Then for all $X,Y \in \mathbb{S}_D$; 	
    \begin{equation}\label{eq:mvi}
        \big| \ln\norm{\phi(X)} - \ln\norm{\phi(Y)}\big| \leq 2\frac{\norm{\phi}}{v(\phi)} \dis{X}{Y} \ ,
    \end{equation}
    with $v(\phi)$ as in eq.\ \eqref{eq:v}.
\end{lemma}
\begin{remark}\label{MVIForAdjoint} 
    For $\phi=\varphi_n^*$, we have $\ker{\phi}\cap\mathbb{S}_D = \emptyset$ with probability one under the Assumption \ref{1}, see \cite[Lemma 2.1]{movassagh2022ergodic}.  Under Assumption \ref{2}, $v(\phi)$ is non-zero with probability $1$ and the right-hand-side of eq.\ \eqref{eq:mvi} is finite almost surely. 
\end{remark}
\begin{proof}
    Let $g:(\mathbb{S}_D,\norm{\cdot}) \to \mathbb{R}$ be defined as $g(X) = \norm{\phi(X)}$. Since $\phi$ is positive with no matrix in $\mathbb{S}_D$ in its kernel we must have that $g(X) > 0$ for all $X \in \mathbb{S}_D$. Since $\mathbb{S}_D$ is compact in the standard topology, we have that 
        \begin{equation}
            v(\phi) = \min\{\norm{\phi(Z)} : Z \in \mathbb{S}_D\}  >  0 \ .
        \end{equation}
    It follows from the mean value inequality, applied to $\ln$, that 
        \begin{equation}
            \big|\ln\norm{\phi(X)} - \ln\norm{\phi(Y)}\big| \leq \frac{\big|\norm{\phi(X)} -\norm{\phi(Y)}\big|}{v(\phi)} \leq \frac{\norm{\phi} \norm{X-Y}}{v(\phi)}
        \end{equation} 
    The results follows from lemma \ref{equvilanceOfMatrics} as $\norm{X-Y} \leq 2\dis{X}{Y}$. 
\end{proof}
	
\subsection{Contraction Coefficient for \texorpdfstring{$\phi$}{}}
For any non-destructive positive map $\phi \in \mapspace$ we define the \textit{contraction coefficient} of $\phi$, denoted $c(\phi)$, as follows:
    \begin{equation}
        \cnum{\phi} = \sup \{\dis{\phi\cdot A}{\phi\cdot B} : A,B \in \mathbb{S}_D \} \ .
    \end{equation}
We have the following properties of the contraction coefficient: 
\begin{lemma}[{\cite[Lemma 3.14]{movassagh2022ergodic}}] \label{propertiesOfC}
    If $\phi \in \mapspace$ be a non-destructive positive map, then
    \begin{enumerate} 
            \item $\dis{\phi\cdot X}{\phi\cdot Y} \leq \cnum{\phi} \dis{X}{Y}$  for all $X,Y \in \mathbb{S}_D$ .
        \item  $\cnum{\phi} \leq 1$ and if $\phi$ is strictly positive then $c(\phi) < 1$. 
        \item If there exist $X,Y$ such that $\phi\cdot X \in \mathbb{S}_D^\circ$ and $\phi\cdot Y \in \mathbb{S}_D \setminus\mathbb{S}_D^o$, 
                then $c(\phi) = 1$.
        \item For any non-destructive positive map $\psi$, we have $\cnum{\phi\circ\psi} \leq \cnum{\phi}\cnum{\psi}$.
        \item If $\phi$ is also non-transient, then  $\cnum{\phi} = \cnum{\phi^*}$.
    \end{enumerate}
\end{lemma}
	
\begin{remark}
    We note that the lemma above is stated slightly differently than \cite[Lemma 3.14]{movassagh2022ergodic}. However a close reading of the proof in \cite{movassagh2022ergodic} shows that the above version holds. 
\end{remark}

Under Assumption \ref{1}, the maps $\Phi^{n}$ defined as in eq.\ \eqref{Phi} become strictly positive in finite time.  As a consequence the following result was proved in \cite{movassagh2022ergodic} using Kingman's sub additive ergodic theorem \cite{kingman1976subadditive,kingman1973subadditive,liggett1985improved}: 
\begin{lemma}[{\cite[Lemma 3.11]{movassagh2022ergodic}}]\label{kappaIsConstant}
    Let $(\varphi_n)_{n\geq1}$ and $\Phi^{n}$ be as in eq.\ \eqref{Phi}. If Assumption \ref{1} holds, then there exists a deterministic constant $\kappa\in[0,1)$ such that almost surely
    \[\ln\kappa = \lim_{n\to\infty} \dfrac{1}{n} \ln\cnum{\Phi^{(n)}}  \]
    and 
    \[\ln\kappa = \lim_{n\to\infty} \dfrac{1}{n} \avg{\ln\cnum{\Phi^{(n)}}} = \inf_{n\in\mathbb{N}}\frac{1}{n} \avg{\ln\cnum{\Phi^{(n)}}} \ .   \]
\end{lemma}
\begin{remark}
    In \cite{movassagh2022ergodic} the ergodic map $\theta$ is assumed to be invertible.  However, a close reading of the proof of \cite[Lemma 3.11]{movassagh2022ergodic} shows that invertibility of $\theta$ is not required.
\end{remark}	
\noindent Lemma \ref{kappaIsConstant} directly yields the following corollary: 	
\begin{cor}\label{cConvergeTo0}
    $\lim_{n\to\infty} c(\Phi^{(n)}) = 0$ almost surely.
\end{cor}

The contraction provided by Lemma \ref{kappaIsConstant} is the driving force behind the convergence $L_n \to Z_1$ state in Lemma \ref{existanceofz}.  In fact this convergence can be made more quantitative:
\begin{lemma}[{\cite[Lemma 3.12]{movassagh2022ergodic}}]\label{lem:quantz} 
    Let $(\varphi_n)_{n\geq1}$ and $\Phi^{(n)}$ be as in eq.\ \eqref{Phi} and suppose that Assumption \ref{1} holds.  Let $L_n$ be as in eq.\ \eqref{LAndR} and let $Z_1= \lim_n L_n$ and $Z_k=Z_1\circ \theta^{k-1}$ be as in Lemma \ref{existanceofz}.  Then, for each $Y\in \mathbb{S}_D$ and $k\in \mathbb{N}$, 
    $$ \dis{(\varphi^*_k\circ\ldots\circ\varphi^*_n)\cdot Y}{Z_k} \ \le \ c(\varphi^*_k\circ\ldots\circ\varphi^*_n) $$
    for all sufficiently large $n$.  In particular, we have $\lim_n(\varphi^*_k\circ\ldots\circ\varphi^*_n)\cdot Y = Z_k $ with probability one.
\end{lemma}
Below it will be useful to consider the contraction obtained from only a fraction of the process. This is described in the following 
	
\begin{lemma}\label{shiftedSubaditiveErgodic}
    Let $(\varphi_n)_{n\geq1}$ and $\Phi^{n}$ be as in eq.\ \eqref{Phi}. Let $\alpha \in (0,1)$ and let $n_\alpha=\lfloor (1-\alpha)n\rfloor$, the integer part of $(1-\alpha)n$. If Assumption \ref{1} holds, then 
        \begin{equation}
            \lim_{n\to\infty} \frac{1}{n} \ln c(\varphi_{n} \circ \cdots \circ \varphi_{n_\alpha + 1})  \ = \    \alpha \ln \kappa \  \text{ almost surely},
        \end{equation}
    where $\kappa$ is the deterministic constant in Lemma \ref{kappaIsConstant}. 
\end{lemma}
\begin{proof}
    First note that, by Part 4 of Lemma \ref{propertiesOfC}, we have 
    \begin{eqnarray}
    \ln\cnum{\varphi_{n}\circ\ldots\circ\varphi_{n_\alpha+1}}\ \geq \ 
        \ln\cnum{\varphi_{n}\circ\ldots\circ\varphi_{1}} -  \ln\cnum{\varphi_{n_\alpha}\circ\ldots\circ\varphi_{1}} \ .
    \end{eqnarray}
    Thus, by Lemma \ref{kappaIsConstant}, 
    \begin{equation}\label{lowerbnd}
        \liminf_{n\to\infty}\frac{1}{n}\ln\cnum{\varphi_{n}\circ\ldots\circ\varphi_{n_\alpha+1}} \ \geq \ \alpha\ln\kappa \quad \text{almost surely.}
    \end{equation}
    
    To prove the complementary upper bound, i.e., that
    \begin{equation}\label{upperbnd}
        \limsup_{n\to\infty}\frac{1}{n}\ln\cnum{\varphi_{n}\circ\ldots\circ\varphi_{n_\alpha+1}} \ \leq \ \alpha\ln\kappa  \ ,
    \end{equation}
    we will show that for each $m\in \mathbb{N}$
    \begin{equation}
        \limsup_{n\to\infty} \frac{1}{n}\ln\cnum{\varphi_{n}\circ\ldots\circ\varphi_{n_\alpha+1}}\ \leq \ \alpha\frac{1}{m}\avg[\ln\cnum{\Phi^{(m)}}] \quad \text{almost surely.}
    \end{equation}
    Eq.\ \eqref{upperbnd} will then follow by Lemma \ref{propertiesOfC}.
    
    Let $m\in\mathbb{N}$ be fixed and consider $n\in\mathbb{N}$ large enough that $n-n_\alpha > 2m$.  Let $p(n) = \lfloor\frac{n_\alpha+m}{m}\rfloor$ and let  $q=q(n) \in \mathbb{N}$ and $r=r(n)\in \ \{0,1,\ldots,m-1\}$ be defined by $n = qm+r$. Then,
    \begin{equation}\label{thereAreEnoughTerms}
        n_\alpha +1 \ \le \ p(n)m+1 \ \le \ n_\alpha + m \ < \ n-m + 1 \le  \ q(n)m \ .
    \end{equation}
    Since $\ln\cnum{\varphi} \le 0$ for any $\varphi\in\mapspace$, we have, using lemma \ref{propertiesOfC}, that
    \begin{equation}
         \ln\cnum{\varphi_{n} \circ \ldots \circ \varphi_{n_\alpha+1}} \ \leq \ \ln\cnum{\varphi_{q(n)m+j}\circ\ldots \circ\varphi_{p(n)m+j+1}} \ 
    \end{equation}
    for any $0\le j\le m-1$, where eq.\ \eqref{thereAreEnoughTerms} guarantees that $p(n)m+j+1\ge 1$ and the composition on the right hand side has non-zero number of factors. Using, \ref{propertiesOfC} again we find that
    $$\ln\cnum{\varphi_{n} \circ \ldots \circ \varphi_{n_\alpha+1}}  \ \le \ \sum_{k=p(n)}^{q(n)-1} \ln\cnum{\varphi_{km+j+m}\circ\ldots\circ\varphi_{km+j+1}} \
    = \ \sum_{k=p(n)}^{q(n)-1} \ln\cnum{\varphi_{m}\circ\ldots\circ\varphi_{1}}\circ\theta^{km+j} \ . $$
    Since this holds for any $j\in\{0, 1, \ldots, m-1\}$, we have 
    \begin{align*}
        \ln\cnum{\varphi_{n} \circ \ldots \circ \varphi_{n_\alpha+1}} \
        &\leq \ \frac{1}{m}\sum_{j=0}^{m-1} \sum_{k=p(n)}^{q(n)-1} \ln\cnum{\varphi_m\circ\ldots\circ\varphi_1}\circ\theta^{km+j}\\
        &=  \ \frac{1}{m} \sum_{i=p(n)m}^{q(n)m-1}\ln\cnum{\varphi_m\circ\ldots\circ\varphi_1}\circ\theta^{i}\\
        &= \  \sum_{i=0}^{q(n)m-1} \frac{1}{m}\ln\cnum{\Phi^{(m)}}\circ\theta^{i} \ - \ \sum_{i=0}^{p(n)m-1} \frac{1}{m} \ln\cnum{\Phi^{(m)}}\circ\theta^{i} \ .
    \end{align*}
    Since $\big(\frac{1}{m}\ln\cnum{\varphi_m\circ\ldots\circ\varphi_1}\big)^+\in L^1(\Omega)$ (where $(\cdot)^+$ denotes the positive part), eq.\ \eqref{upperbnd} follows from the Birchoff ergodic theorem. 
\end{proof}

\subsection{Invertible ergodic dynamics}
In this section, we assume that $\theta$ is an invertible ergodic map. It is often possible to replace the original dynamical system by a natural extension on which $\theta$ is invertible; for instance this is possible if $\theta$ is \emph{essentially surjective}, i.e. if $\Omega \setminus\theta(\Omega)$ is a sub-null set \textemdash see \cite{cornfeld2012ergodic}. We will denote this extension also by $(\Omega, \mathcal{F}, \prob, \theta)$ and note that the previously stated results still hold. 
  
Since $\theta$ is invertible and measure preserving, the inverse map  $\theta^{-1}$ is also a measure preserving ergodic transformation. We extend the definition of $(\Phi^{(n)})$ to include negative indices as follows
\begin{equation}\label{newPhi}
    \Phi^{(n)} (\omega)= 
    \begin{cases}
        \varphi_n(\omega) \circ \ldots \circ \varphi_1(\omega) & \text{ for } n \geq 1 \ ,\\
        \varphi_0 & \text{ for } n = 0 \ ,\\
        \varphi_{-1}(\omega) \circ \ldots \circ \varphi_{n}(\omega) & \text{ for } n \leq -1 \ ,
    \end{cases}
\end{equation}
where $\varphi_{n} := \varphi_{\theta^{n}} $ for all $n$. When $\theta$ is invertible, Assumption \ref{1} guarantees that with probability one $(\Phi^{(-n)})_{n\geq 1}$ is almost surely eventually strictly positive \textemdash \ see \cite[Lemma 3.13]{movassagh2022ergodic}. 
 
With this extended dynamical system, we introduce some new notation. Let $n \in \mathbb{N}$ and define
\begin{equation}\label{eq:Psin}
    \psi_n = \varphi_{-n}^* \quad \text{ and } \quad \Psi^{(n)} = \psi_n \circ \ldots \circ \psi_1 \ .
\end{equation}
Note that $\Psi^{(n)*} = \Phi^{(-n)}$.  We see that $(\Psi^{(n)})_{n\in\mathbb{N}}$ is almost surely eventually strictly positive. This allows us to define a new stopping time $\tau'$ as:
\begin{equation}
    \tau' \ =  \ \inf \{n \geq 1: \Phi^{(n+k)} \text { and } \Psi^{(n+k)} \text{ are strictly positive } \forall k \geq 0\} \ , 
\end{equation}
satisfying $\prob[\tau' <\infty]=1$ if $\theta$ is invertible and Assumption \ref{1} holds.

We have the following result analogous to Lemma \ref{kappaIsConstant} for the sequence $(\Psi^{(n)})_{n\geq1}$: 	
\begin{lemma}\label{contractionCoefficientPsi} 
    If $\theta$ is invertible and $(\Phi_{(n)})_{n\geq1}$ satisfies Assumption \ref{1}, then 
    \begin{equation}\label{eq:kappaPsin}
        \ln \kappa  \ \overset{a.s}= \ \lim_{n\to\infty} \frac{1}{n} \ln \cnum{\Psi^{(n)}},
    \end{equation}
    where $\Psi^{(n)} $ is as in eq.\ \eqref{eq:Psin} and $\kappa$ is the deterministic constant appearing in lemma \ref{kappaIsConstant}.  In particular, $\lim_n c(\Psi^{(n)})=0$ almost surely.
\end{lemma}
\begin{remark} 
    The existence of the deterministic limit on the right hand side of eq.\ \eqref{eq:kappaPsin} follows directly from Lemma \ref{kappaIsConstant} applied with the sequence $\Psi^{(n)}$ in place of $\Phi^{(n)}$.  That the limit equals $\kappa$ follows from the identity
    $$ \avg{ \ln \cnum{\Psi^{(n)}} } = \avg{ \ln \cnum{\Phi^{(n)*}_{\theta^{-n-1}}}} = \avg{ \ln \cnum{\Phi^{(n)}}} \ ,$$
    where we have used the facts that $\theta$ is measure preserving and that $\cnum{\phi^*}=\cnum{\phi}$ for any $\phi$.
\end{remark}

If $\theta$ is invertible and Assumption \ref{1} holds, then the left and right Perron-FRobenius eigenmatrices $R_n$ and $L_n$ for $\Phi^{(n)}$ exist also for large negative $n$.  As a result we have the following lemma for the convergence of the right eigenvectors: 
\begin{lemma}[{\cite[Lemma 3.14]{movassagh2022ergodic}}]\label{existanceofz'}
    Let $(\varphi_n)_{n\geq1}$ and $\Phi^{n}$ be as in eq.\ \eqref{Phi} and let $R_n$ be the right Perron-Frobenius eigenmatrix for $\Phi^{(n)}$, see eq.\ \eqref{LAndR}. If $\theta$ is invertible an Assumption \ref{1} holds, then there is an $\mathbb{S}_D^\circ$ valued random variable $Z_1'$ such that
    \begin{equation}
        \lim_{n\to-\infty} R_n \overset{a.s} = Z_1' \ 
    \end{equation}
    and, with $Z'_k := Z'_1 \circ \theta^{-k+1}$, we have:
    \begin{enumerate}
        \item for every $k \in \mathbb{N}$, $\psi^*_k\cdot Z_{k+1}' = Z_k'$ a.s., and
        \item for each $Y\in \mathbb{S}_D$ and $k\in \mathbb{N}$, 
    $$ \dis{(\psi^*_k\circ\ldots\circ\psi^*_n)\cdot Y}{Z_k'} \ \le \ c(\psi^*_k\circ\ldots\circ\psi^*_n) $$
    for all sufficiently large $n$.  In particular, we have $\lim_{n} (\psi^*_k \circ \ldots \circ \psi^*_n)\cdot Y = Z'_k $ a.s..
    \end{enumerate}			
\end{lemma}

If instead we take $n\rightarrow \infty$, we do not have almost sure convergence of $R_n$.  However, we do have convergence in distribution:
\begin{cor}
    We have that 
    \begin{equation}
        R_n \xrightarrow[n\to\infty]{d} Z'_1 \quad \text{and}\quad   L_n \xrightarrow[n\to-\infty]{d} Z_1 \ ,
    \end{equation}
    where $\xrightarrow{d}$ denotes convergence in distribution.
\end{cor}
\begin{proof} 
    Note that $R_n=R_{-n;\theta^{n+1}}$, so that $R_n \overset{d}=R_{-n}$.  Since $\lim_{n\to \infty}R_{-n}=Z_1'$ a.s., the first limit holds. The proof for the second limit is similar. 
\end{proof}

\section{Proof of the Law of Large Numbers}\label{sec:lln} 

We now describe the proof of Theorem \ref{lln}. Recall from the discussion following the statement of the theorem above, that it suffices to prove eq.\ \eqref{eq:firstpart}, which states that
$$\lim_{n\to \infty} \sup_{X,Y \in \mathbb{S}_D} \Big | \frac{1}{n}\ln \inner{Y}{\Phi^{(n)}(X)} - l \Big | \ = \ 0 \quad \text{a.s.} .$$
To this end, note that by Assumption \ref{2} we have $\avg [\ln \norm{\varphi^*_k(Z_{k+1})}] < \infty$ for each $k \in \mathbb{N}$. Thus by Birkhoff's ergodic theorem we have 
$$
    \lim_{n\to\infty} \frac{1}{n} \sum_{k=1}^n \ln \norm{\varphi^*_k(Z_{k+1})}\ \overset{a.s}= \ \avg \ln \norm{\varphi^*_0(Z_1)} \ := \ l.
$$
Thus eq.\ \eqref{eq:firstpart}, and therefore Theorem \ref{lln}, follows from the following
\begin{lemma}\label{En&Dn}
    Suppose that Assumption \ref{1} holds and let
    \begin{equation}\label{seqDn}
        D_n \ = \ \sup_{X,Y \in \mathbb{S}_D} \big\{\big| \ln \inner{Y}{\Phi^{(n)}(X)} - \ln \norm{\Phi^{(n)*}(Y)}\big|\big\} \ ,
    \end{equation}
    and
    \begin{equation}\label{seqEn}
        E_n \ := \ \sup_{Y \in \mathbb{S}_D} \Big\{\Big| \ln \norm{\Phi^{(n)*}(Y)} - \sum_{k=1}^n \ln \norm{\varphi^*_{k,\omega}(Z_{k+1})} \Big|\Big\} \ .
    \end{equation} 
    for $n\geq 1$.  Then, with probability one, 
    \begin{enumerate}
        \item $D_n$ is eventually bounded, i.e., $\displaystyle{\limsup_{n\to\infty}} D_n < \infty$, and
        \item $\displaystyle{\lim_{n\to \infty}} \tfrac{1}{n} E_n \ = \ 0$.
    \end{enumerate}    
\end{lemma} 
\begin{remark} 
    Note that from $\lim_n \frac{1}{n}E_n$ we conclude directly that $\lim_{n} \frac{1}{n} \ln \norm{\Phi^{(n)*}(Y)} = l$ for every $Y\in \mathbb{S}_D$.  In particular $l=\lambda$, the top Lyapunov exponent of $\Phi^{(n)}$, as claimed in the discussion following Theorem \ref{lln} above.
\end{remark}

\begin{proof}
    First note that for any $X,Y\in \mathbb{S}_D$,
    \begin{equation}\label{innerProductUpperBound}\inner{Y}{\Phi^{(n)}(X)} \ = \ \inner{\Phi^{(n)*}(Y)}{X} \ \le \ \inner{\Phi^{(n)*}(Y)}{\mathbb{I}} \ = \ \norm{\Phi^{(n)*}(Y)} \ ,\end{equation}
    Here we have used that $X\le \mathbb{I}$ and $\tr M = \norm{M}$ for any positive semi-definite matrix.

    For the rest of the proof, we restrict to a configuration $\omega$ such that $\tau=\tau_\omega < \infty$.  Such configurations form a full measure set by Assumption \ref{1}.
    
    Because $\Phi^{(\tau)}$ is strictly positive, we have $\min \sigma(\Phi^{(\tau)}(P)) > 0 $ for any projection $P$, where $\sigma(\Phi^{(\tau)}(P))$ denotes the spectrum of $\Phi^{(\tau)}(P)$.  Thus the map $P\mapsto \min \sigma(\Phi^{(\tau)}(P)) $ is a continuous function from the set of rank-$1$ projections into $(0,\infty)$.  Since the set of rank-$1$ projections is compact, we have
    $$ a \ := \ \min \left \{ \min (\sigma(\Phi^{(\tau)}(P))) \ : \ P \text{ is a rank-$1$ projection} \right \} \ > \ 0 .$$

    Given $X,Y\in \mathbb{S}_D$ and $n>\tau$, let $W= \varphi^*_{\tau+1}\circ\ldots\circ\varphi^*_n (Y)$. Because $X$ has at least one eigenvalue greater than or equal to $\tfrac{1}{D}$, we have $X\ge \tfrac{1}{D}P$ for some rank-$1$ projection $P$, and thus
    \begin{equation*}
        \inner{Y}{\Phi^{(n)}(X)} \ = \ \inner{W}{\Phi^{(\tau)}(X)} \ \geq \ \frac{1}{D} \inner{W}{\Phi^{(\tau)}(P)} \ \geq \ \frac{a}{D} \inner{W}{\mathbb{I}}  \ =  \ \frac{a}{D} \norm{W} \ .
    \end{equation*}
    Since $\norm{\Phi^{(n)*}(Y)} \ = \norm{\Phi^{(\tau)*}(W)} \ \le \ \norm{\Phi^{(\tau)*}} \ \norm{W}$,
    we have
    \begin{equation}\label{innerProductLowerBound}
        \inner{Y}{\Phi^{(n)}(X)} \ \ge \frac{a}{D \norm{\Phi^{(\tau)*}}} \norm{\Phi^{(\tau)*}(Y)} \ .
    \end{equation}
    Putting  eqs.\ \eqref{innerProductUpperBound} and \eqref{innerProductLowerBound} together, we see that 
    $$ \ln a -\ln D -\ln \norm{\Phi^{(\tau)*}} \ \le \ \ln\inner{Y}{\Phi^{(n)}(X)} - \ln \norm{\Phi^{(\tau)*}(Y)} \ \le \ 0 $$
    for $X,Y\in \mathbb{S}_D$ and $n>\tau$. It follows that $\limsup_n D_n \le \ln D + \ln \norm{\Phi^{(\tau)*}} - \ln a < \infty$ whenever $\tau < \infty$.

    Turning now to the proof that $\lim_n \frac{1}{n}E_n = 0 $, consider $n >\tau$.  Note that 
    $$\norm{\Phi^{(n)*}(Y)} \ = \ \norm{\phi_1^*( \phi_{2}^*\circ \cdots \circ \phi_n^*(Y))} \ = \ \norm{\phi_1^*\left((\phi_{2}^*\circ \cdots \circ \phi_n^*)\cdot Y \right)} \, \norm{\phi_{2}^*\circ \cdots \circ \phi_n^*(Y)} \ ,$$
    where in the final expression we have introduced the projective action by multiplying and dividing by $ \norm{\phi_{2}^*\circ \cdots \circ \phi_n^*(Y)}  = \tr \phi_{2}^*\circ \cdots \circ \phi_n^*(Y)$.  Taking logarithms and iterating, we find that
    \begin{equation*}
        \ln \norm{\Phi^{(n)*}(Y)} \ = \ \sum_{k=1}^{n} \ln \norm{\varphi_k^*((\varphi_{k+1}^*\circ\ldots\circ\varphi_n^*)\cdot Y)} \ ,
    \end{equation*}
    where the empty composition $\varphi_{n+1}^*\circ\ldots\circ \varphi_n^*$ is understood as the identity map.
    Thus
    \begin{multline}
    \label{eq:firststepEn}
        E_n(Y) \ := \ \big|\ln\norm{\Phi^{(n)*}(Y)} - \sum_{k=1}^n \ln \norm{\varphi^*_k(Z_{k+1})}\big| \ \le \ \sum_{k=1}^n E_n^k(Y) \ , \\
        \text{ with } \quad E_n^k(Y) \ := \  \big | \ln \norm{\varphi_k^*((\varphi_{k+1}^*\circ\ldots\circ\varphi_n^*)\cdot Y)} - \ln \norm{\varphi^*_k(Z_{k+1})}\big| \ .
    \end{multline}
    Using Lemma \ref{meanValueIneqaulity}, Remark \ref{MVIForAdjoint}, and Lemma \ref{lem:quantz} we may bound $E_n^k(Y)$ as follows
    \begin{equation}\label{eq:EnkY}
        E_n^k(Y) \  \le \ 2 \frac{\norm{\varphi^*_k}}{v(\varphi^*_k)} \cnum{\varphi_{k+1}^*\circ\ldots\circ\varphi_n^*} \ .
    \end{equation}
    
    Now let $\alpha\in (0,1)$ and let $n_\alpha$ be the integer part of $(1-\alpha)n$.  We will bound the terms on the right hand side of \eqref{eq:firststepEn} differently according to if $k< n_\alpha$ or $k \ge n_\alpha$.  For $k < n_\alpha$, we have 
    \[ E_n^k(Y) \ \le \ 2 \frac{\norm{\varphi^*_k}}{v(\varphi^*_k)} \cnum{\varphi_{n_\alpha}^*\circ\ldots\circ\varphi_n^*} \ ,
    \]
    where we have used eq.\ \eqref{eq:EnkY} and applied Lemma \ref{propertiesOfC} to bound $\cnum{\varphi_{k+1}^*\circ\ldots\circ\varphi_n^*}\le \cnum{\varphi_{n_\alpha}^*\circ\ldots\circ\varphi_n^*}$.
    For $k\ge n_\alpha$, on the other hand, we have 
    \[ E_n^k(Y) \ \le \ \big |\ln \norm{\varphi_k^*((\varphi_{k+1}^*\circ\ldots\circ\varphi_n^*)\cdot Y)}\big | + \big | \ln \norm{\varphi^*_k(Z_{k+1})}  \big |
    \ \le \ 2 \Big ( \big | \ln v(\varphi_k^*)\big | + \big | \ln \norm{\varphi_k^*} \big | \Big ) \ .
    \]
    Thus 
    \begin{equation}\label{eq:Enbound}
        E_n \ = \ \sup_{Y\in \mathbb{S}_D} E_n(Y) \ \le \ S^<_n  + S^>_n
    \end{equation}
    with 
    \begin{align*}
        S^<_n \ =& \  2 \sum_{k=1}^{n_\alpha-1} \frac{\norm{\varphi^*_k}}{v(\varphi^*_k)} \cnum{\varphi_{n_\alpha}^*\circ\ldots\circ\varphi_n^*}  \ , 
        \intertext{and} 
        S^>_n \ =& \   2 \sum_{k=n_\alpha}^{n} \Big ( \big | \ln v(\varphi_k^*)\big | + \big | \ln \norm{\varphi_k^*} \big | \Big )\ .
    \end{align*}
    We will prove that $\lim_n S^<_n = 0$ and $\lim_n \frac{1}{n}S^>_n=O(\alpha)$. 
    
    Note that by Assumption \ref{2} we have $\avg[|\ln(\frac{ \norm{\varphi_0^*}}{ v(\varphi^*_0)})|] <\infty  $.  Thus, for any $\epsilon > 0$,
    \begin{align*}
        \infty \ > \ \frac{1}{\epsilon} \avg[|\ln(\tfrac{ \norm{\varphi_0^*}}{ v(\varphi^*_0)})|] 
            \ \ge& \  \sum_{k=1}^\infty  \prob \big( \ln(\tfrac{ \norm{\varphi_0^*}}{ v(\varphi^*_0)}) > k\epsilon\big) \\
        =& \ \sum_{k=1}^\infty \prob \big( \ln(\tfrac{ \norm{\varphi_0^*}}{ v(\varphi^*_0)}) >k\epsilon\big) \ 
            =  \ \sum_{k=1}^\infty \prob\big( \tfrac{\norm{\varphi_k^*}}{v(\varphi_k^*)} > e^{k\epsilon}\big) \ .
    \end{align*}
    Hence, by the Borel-Cantelli Lemma, we find that $\limsup_{k} e^{k\epsilon} \tfrac{\norm{\varphi_k^*}}{v(\varphi_k^*)}  \le 1 $ with probability one.  Taking $\epsilon <  \alpha |\ln \kappa|$, we conclude from Lemma \ref{shiftedSubaditiveErgodic} that
    $$\limsup_{n\to \infty} S_n^< \ \le \ \limsup_{n\to \infty} n_\alpha e^{\epsilon n_\alpha} \cnum{(\varphi_{n_\alpha}^*\circ\ldots\circ\varphi_n^*)} \ = \ 0 \ .$$
    In particular, we also have $\lim_n \frac{1}{n} S_n^< = 0$.
    
    Now consider $S_n^>$.  Since $\ln v(\phi_0^*) $ and $\ln \norm{\phi_0^*}$ are $L^1$ random variables by Assumption \ref{1}, we conclude from the Birkhoff ergodic theorem \cite{birkhoff1931proof} that
    $$\lim_n \frac{1}{n} S_n^> \ = \ 2\alpha  \left [\avg |\ln\norm{\varphi_0^*}| + \avg |\ln v(\varphi_0^*)|\right ] \ . $$
    We conclude that $\limsup_n \frac{1}{n}E_n \ = \ O(\alpha)$.  Since $\alpha \in (0,1)$ was arbitrary, we have $\lim_n \frac{1}{n}E_n=0$.
\end{proof}

\section{Proof the Central Limit Theorem}\label{sec:proofCLT}
In this section we prove Theorem \ref{CLT}.  Let $(X_n)_{n\geq1}$ and $(Y_n)_{n\geq1}$ be sequences in $\mathbb{S}_n$.  Then
\begin{multline*}
     \frac{1}{\sqrt{n}}\big( \ln \inner{Y_n}{\Phi^{(n)}(X_n)} - nl\big) \
     = \ \frac{1}{\sqrt{n}} \big (\ln \inner{Y_n}{\Phi^{(n)}(X_n)}- \ln \norm{\Phi^{(n)*}(Y_n)} \big ) \\
     +  \frac{1}{\sqrt{n}}\big(\ln \norm{\Phi^{(n)*}(Y_n)}- \sum_{k=1}^n\ln \norm{\varphi^*_k(\omega)(Z_{k+1}(\omega))} \big )
        + \dfrac{1}{\sqrt{n}}\sum_{k=1}^n\xi_k \ ,
\end{multline*}
where $\xi_k = \ln \norm{\varphi^*_k(Z_{k+1})} - l$. Thus
\begin{equation*}
  \left |   \frac{1}{\sqrt{n}}\big( \ln \inner{Y_n}{\Phi^{(n)}(X_n)} - nl\big) - \dfrac{1}{\sqrt{n}}\sum_{k=1}^n\xi_k \right | \ \le \ \frac{1}{\sqrt{n}} (D_n + E_n) 
\end{equation*}
with $D_n$ and $E_n$ as in eqs.\ \eqref{seqDn} and \eqref{seqEn}, respectively.
By Lemma \ref{En&Dn}, $D_n$ is almost surely eventually bounded.   Thus to prove that $\big( \frac{1}{\sqrt{n}} \ln \inner{Y_n}{\Phi^{(n)}(X_n)}\big )_{n\ge 1}$ converges in distribution to a centered normal variable, it suffices to prove the following two results:
\begin{enumerate}
    \item $\tfrac{1}{\sqrt{n}} E_n$ converges to $0$ in probability, and
    \item $Q_n := \displaystyle{\frac{1}{\sqrt{n}}\sum_{k=1}^n} \xi_k$ converges in distribution to a centered normal variable with variance given by eq.\ \eqref{sigma} above.
\end{enumerate}
These results are proved in Lemma \ref{sqrtEn} and Lemma \ref{MartingaleApproxLemma} below, respectively.  
		
\begin{lemma}\label{sqrtEn}
    Suppose that $\theta$ is invertible and that Assumption \ref{1} holds. Let $(E_n)_{n=1}^\infty$ be the variables defined in eq.\ \eqref{seqEn}. Then $(E_n)_{n=1}^\infty$ is tight. In particular, $(\frac{1}{\sqrt{n}} E_n)_{n=1}^\infty$ converges to $0$ in probability.
\end{lemma} 
\begin{proof}
    Following the proof of eq.\ \eqref{eq:Enbound} above, but applying in the proof of Lemma \ref{En&Dn}, we have
    \[
        E_n \ \le \ S_n \ := \ 2\sum_{k=1}^n \frac{\norm{\varphi_{k}^*}}{v(\varphi^*_{k})}\cnum{\varphi^*_{k+1} \circ \ldots \circ \varphi^*_n} \ .
    \]
    We prove that $E_n$ are tight by showing that $S_n \overset{d}{=} S_n'$ where the random variables $S_n'$ satisfy $\sup_n S_n' <\infty$ almost surely.
    
    Consider the variables $S_n'=S_{n;\theta^{-n}}$.  Since $\cnum{\phi^*}=\cnum{\phi}$, we have
    $$
        S_n' \ = \ 2\sum_{k=0}^{n-1} \frac{\norm{\varphi_{-k}^*}}{v(\varphi^*_{-k})}\cnum{\varphi_0\circ \ldots \circ \varphi_{1-k}} \ .
    $$
    As above the empty composition appearing at $k=0$ is understood as the identity map. By the Borel-Cantelli similar to that used to bound $S_n^<$ in the proof of Lemma \ref{En&Dn} we see that 
    $$ \limsup_{k\to \infty} \frac{1}{k} \ln \left ( \tfrac{\norm{\varphi_{-k}^*}}{v(\varphi^*_{-k})}\right )  \ = \ 0 \ \quad \text{a.s.}. $$
    On the other hand by Lemma \ref{contractionCoefficientPsi} we have
    $$ \limsup_{k \to \infty} \frac{1}{k} \ln \cnum{\varphi_{0} \circ \ldots \circ \varphi_{1-k}} \ = \ \ln \kappa \ < \ 0 \quad \text{a.s.}.$$
    It follows that 
    $$\lim_{n\to \infty} S_n' \ = \ 2\sum_{k=0}^{n} \frac{\norm{\varphi_{-k}^*}}{v(\varphi^*_{-k})}\cnum{\varphi_0\circ \ldots \circ \varphi_{1-k}} \ =: \ S_\infty'$$
    is finite almost surely. Clearly $S_\infty'=\sup_n S_n'$.

    Since $S_n \overset{d}{=}S_n'$ we have
    $$\prob \left [ E_n > \epsilon \right ] \ \le \ \prob \left [ S_n > \epsilon \right ] \ = \ \prob \left [ S_n' > \epsilon \right ] \ \le \ \prob \left [ S_\infty' > \epsilon \right ] \ ,$$
    so $(E_n)_{n=1}^\infty$ is tight as claimed.  It follows that
    $$\prob \left [ \tfrac{1}{\sqrt{n}} E_n > \epsilon \right ] \ \le \ \prob \left [ S_\infty' > \sqrt{n} \epsilon \right ] \ \to \ 0 \ ,$$
    so $ \tfrac{1}{\sqrt{n}} E_n$ converges to zero in probability.
\end{proof}
		
To prove the convergence of $Q_n = \frac{1}{\sqrt{n}}\sum_{k=1}^n \xi_k$  to a centered normal law in distribution, we use the martingale approximation method of Gordin \cite{gordin1969central}.  The following proof is adapted from the proof of \cite[Lemma 9.2]{hennion1997limit} and is similar to the proof of \cite[Theorem 1.1]{liverani1996central}. The key idea is to find a \emph{reverse martingale difference} with respect to the filtration $(\mathcal{F}^n)_{n\geq1}$ and use the Central Limit Theorem for (reverse) martingale differences \cite{billingsley2008probability, brown1971martingale, conze2007limit} which was proved independently by Billingsly \cite{billingsley1961lindeberg} and  \cite{ibragimov1963central} for the ergodic case:

\begin{namedtheorem}[Martingale Difference Central Limit Theorem]
    Let $(X_n)_{n\geq1}$ be a stationary ergodic direct or reversed martingale difference with respect to a filtration $\{\mathcal{A}_n\}_{n\geq1}$. If $X_1 \in L^2$, then $\big(\frac{1}{\sqrt{n}}\sum_{k=1}^n X_k \big)_{n\ge 1}$ converges in distribution to a centered normal random variable with variance $\sigma^2 = \avg(X_1^2)$.
\end{namedtheorem}
		
\begin{lemma}\label{MartingaleApproxLemma}
    Suppose that $\theta$ is invertible, that Assumption \ref{1} holds, and Assumption \ref{2'} holds for some $p\ge 2$.
    Let $\xi_k = \ln \norm{\varphi^*_k(Z_{k+1})} - l$ for $k\in \mathbb{Z}$.  If 
    \begin{equation}\label{summability}
        \sum_{n=1}^\infty \norm{\cond{\xi_0}{n}}_q < \infty	\ ,
    \end{equation}
    with $\frac{1}{p}+\frac{1}{q}=1$, then the sequence $(Q_n)_{n=1}^\infty$ given by
    \begin{equation}\label{eq:Qn}
          Q_n \ = \  \frac{1}{\sqrt{n}}\sum_{k=1}^n  \xi_k 
    \end{equation}
    converges in distribution to a centered normal law with variance $\sigma^2 < \infty$. Furthermore $\sigma = 0$ if and only if there exists stationary sequence $(g_n)_{n\geq 1}$ such that 	
    \begin{equation}
        g_n \in L^q(\mathcal{F}^n) \quad \text{ and } \quad \xi_n = g_{n+1} - g_n
    \end{equation}
\end{lemma}
\begin{proof}
    Let $M:=  \sum_{k=1}^\infty\norm{\cond{\xi_0}{k}}_q <\infty$ by eq.\ \eqref{summability}. We define
    \begin{equation}\label{eq:g}
        g_0 \ := \ \sum_{k=1}^\infty \cond{\xi_{-k}}{0} \ ,
    \end{equation}
    and note that 
    $$\norm{g_0}_q \ \le \ \sum_{k=1}^\infty \norm{\cond{\xi_{-k}}{0}} \ = \  M \ , $$
    since $\theta$ is measure preserving.  Since $\norm{\cdot}_1 \le \norm{\cdot}_q$, the series defining $g_0$ converges in $L^1$ and hence absolutely, almost everywhere.  
    
    We define
    \begin{equation}\label{zetaAndG}
        \zeta_0 \ = \ \sum_{k=0}^\infty \big( \cond{\xi_{-k}}{0} - \cond{\xi_{-k}}{1}\big)\ ,
    \end{equation}
    and note that $ \zeta_0 \ = \ \xi_0 + g_0 - g_0\circ \theta$.  For $n\in \mathbb{Z}$, we now define $\zeta_n = \zeta_0 \circ \theta^n$ and $g_n = g_0 \circ \theta^n$, so that
    \begin{equation}\label{martingaledifferenceapprox}
        \xi_n \ = \ \zeta_n + g_{n+1} - g_n \ .
    \end{equation}
    Since 
    \begin{equation}
        \label{eq:xibound} |\xi_n| \ \le \ \ln \norm{\phi_n^*} + \ln v(\phi_n^*) + |l| \ , 
    \end{equation}
    we have $\xi_n\in L_p\subset L_q$ by Assumption \ref{2'},  so $\zeta_n=\xi_n - g_{n+1}+g_n\in L_q \subset L^1$. 
    Taking conditional expectation with respect to $\mathcal{F}^{n+1}$ in eq.\ \eqref{zetaAndG}, we see that
    \begin{equation}
        \cond{\zeta_n}{n+1} = 0 \ ,
    \end{equation}
    i.e., $(\zeta_n)_{n\geq1}$ is a reverse martingale difference (\emph{reverse} because $(\mathcal{F}^n)_{n\ge 1}$ is a reverse filtration). 
    Now eq.\ \eqref{martingaledifferenceapprox} shows  that 
    \begin{equation}
        \frac{1}{\sqrt{n}} \sum_{k=1}^n \xi_k \ = \ \frac{1}{\sqrt{n}} \sum_{k=1}^n \zeta_n + \frac{1}{\sqrt{n}} (g_{n+1} - g_1) \ .
    \end{equation}
    Since $g_{n+1}=g_{1}\circ \theta^n$, we see that $g_{n+1}-g_1$ is tight and thus $\frac{1}{\sqrt{n}} (g_{n+1} - g_1)$ converges to $0$ in probability.  Therefore, by the Martingale Difference Central Limit Theorem, we will have the required convergence in distribution if we establish that $\zeta_0 \in L^2$.  

    Since $\zeta_0=\xi_0  - (g_1-g_0)$ and $\xi_0\in L_p \subset L_2$ by eq.\ \eqref{eq:xibound}, it suffices to show that $g_1-g_0 \in L^2$.
    We have $g_n \in L^q(\mathcal{F}^n)$, but this does not suffice as $q <2$.  To show that $g_1-g_0\in L^2$ we need to exploit cancellation between the two terms.  To this end, let $\lambda \in (0,1)$ and define
    \begin{equation}\label{eq:glambda0}
        g^\lambda_0 \ = \ \sum_{k=1}^\infty \lambda^{k-1}\cond{\xi_{-k}}{0} \ , 
    \end{equation}
    and define $ g^\lambda_n = g^\lambda_0 \circ \theta^n$ for $n\in \mathbb{Z}$.
    Since $\norm{\cond{\xi_{-k}}{0}}_p \le \norm{\xi_{-k}}_p = \norm{\xi_0}_p$, the convergence factor $\lambda^{k-1}$ in eq.\ \eqref{eq:glambda0} guarantees that $g^\lambda_0\in L^p\subset L^2$.  Furthermore, we have
    \begin{equation}\label{eq:glambda0bound} \norm{g^\lambda_0}_q \ \le \ \sum_{k=1}^\infty \lambda^{k-1} \norm{\cond{\xi_{-k}}{0}}_q \ \le \ M \ , \end{equation}
    since $\lambda \le 1$.

    We will now show that $\norm{g_1^\lambda - \lambda g_0^\lambda}_2^2$ is bounded uniformly in $\lambda$.  We start with the estimate
    \begin{multline*}
        \norm{g_1^\lambda - \lambda g_0^\lambda}_2^2 \ = \ (1+\lambda^2) \norm{g_1^\lambda}^2_2 - 2\lambda\avg[g^\lambda_0 g^\lambda_1]  \\
	\leq \ 2\big[\norm{g_1^\lambda}^2_2-\lambda\avg[g^\lambda_0 g^\lambda_1]\big] \
        = \ 2 \avg \left [ g_1^\lambda\left ( g_1^\lambda- \lambda \cond{g_0^\lambda}{1} \right ) \right ] \ ,
    \end{multline*}
    where we have noted that $\norm{g^\lambda_1}_2 = \norm{g^\lambda_0}_2$ (since $(g^\lambda_n)_{n=1}^\infty$ is stationary) and that $g^\lambda_1$ is $\mathcal{F}^1$ measurable. Note that
    \begin{equation*}
        g^\lambda_1 - \lambda\cond{g_0^\lambda}{1} \ = \ \sum_{k=1}^\infty\lambda^{k-1} \cond{\xi_{-k+1}}{1} - \lambda \sum_{k=1}^\infty\lambda^{k-1}\cond{\xi_{-k}}{1} \ = \ \cond{\xi_0}{1} \ .
    \end{equation*}
    Thus
    \begin{equation*}
        \norm{g_1^\lambda - \lambda g_0^\lambda}_2^2 \ \leq \  2  \int_{\Omega} \cond{\xi_0}{1} g_1^\lambda d\prob \ \leq \ 2 \norm{\cond{\xi_0}{1}}_p \norm{ g_1^\lambda }_q \ \leq \ 2\norm{\xi_0}_p M  \ ,
    \end{equation*}
    where we have used H\"older's inequality and eq.\ \eqref{eq:glambda0bound}.

    Since $g_1-g_0 = \lim_{\lambda \uparrow 1} g_1^\lambda -\lambda g_0^\lambda $, we have
    \[  	
        \avg[(g_1 - g_0)^2] \ = \ \avg\big[ \lim_{\lambda\uparrow 1} (g_1^\lambda - \lambda g_0^\lambda)^2\big]
            \ \leq \  \liminf_{\lambda\uparrow 1} \avg[(g_1^\lambda - \lambda g_0^\lambda)^2]  \ \le \ 2 \norm{\xi_0}_p M \ ,
    \]
    by	Fatou's Lemma.  Therefore $g_1-g_0 \in L^2$.  Thus $\zeta_n\in L^2$ for each $n$ and the martingale difference central limit theorem implies that
    $(\frac{1}{\sqrt{n}}\sum_{k=1}^n \zeta_k )_{n\ge 1}$ (and thus
    $(\frac{1}{\sqrt{n}}\sum_{k=1}^n \xi_k )_{n\ge 1}$) converges in distribution to a centered normal random variable with variance $\sigma^2 = \avg[\zeta_0^2]$. 
    
    If $\sigma = 0$ then we have that $\zeta_n=0$ a.s.\ for each $n\in\mathbb{Z}$. In this case, we have $\xi_n = g_{n+1} - g_n$ for the stationary processes $(g_n)_{n\in\mathbb{Z}}$ defined above. This concludes the proof of lemma \ref{MartingaleApproxLemma}
\end{proof}

This completes the proof of Theorem \ref{CLT}.  In the next section we discuss the mixing conditions sufficient to prove the hypothesis eq.\ \eqref{qNormSummable}.

\section{Mixing Conditions}\label{sec:mixing}

In this section we prove Theorem \ref{customthm3}, which provides sufficient conditions for the main hypothesis eq.\ \eqref{qNormSummable} of Theorem \ref{CLT}.
The arguments in this section are based on similar results in \cite{durr1986remarks} and \cite{hennion1997limit}. We rely on the following estimate on averages of sub-multiplicative random variables that combines \cite[Lemma 6.2 \& Lemma 6.3]{hennion1997limit} \textemdash \ see also \cite[Lemma 3 \& Lemma 4]{cohn1993weak}.

\begin{lemma}[{\cite{hennion1997limit}}]\label{boundsForSubMultiplicativeAverages}
    Consider a probability space $(\Omega,\mathcal{F},\prob)$ with an ergodic measure preserving map $\theta:\Omega \to \Omega$,  a filtration $(\mathcal{F}_n)_{n\ge 0}$, and a reverse filtration $(\mathcal{F}^n)_{n\ge 0}$, such that $\theta^{-1}(\mathcal{F}_{n+1})=\mathcal{F}_n$ and $\theta^{-1}(\mathcal{F}^{n+1})=\mathcal{F}^n$ for each $n \ge 0$. Let $\alpha_n$ and $\rho_n$ be mixing coefficients defined as defined in eqs.\ \eqref{alphan} and     \eqref{rhon}, respectively.
    Let $(M_n)_{n\geq 1}$ be a  sequence of $[0,1]$-valued random variables with the following \emph{sub-multiplicative property}
    \begin{equation}\label{eq:submult}
        M_{m+n} \ \leq \ M_m \ M_n\circ\theta^n. 
    \end{equation}
    If for each $0\leq m < n$ it holds that $M_{n-m} \circ \theta^m$ is both $\mathcal{F}_n$ and $\mathcal{F}^m$ measurable, then we have:
    \begin{enumerate}
        \item If $\alpha_n \leq cn^{-\lambda}$ with $c,\lambda >0$, then $\avg[M_n]$ almost vanishes to order $n^{-\lambda}$,
            \begin{equation}\label{eq:firstMnbound}
                \avg[M_n] \ = \ O \bigg(\frac{a_n}{n} \bigg)^\lambda \ 
            \end{equation} 
            for any sequence $(a_n)_{n\geq 1}$ of real numbers such that	
            \begin{equation}\label{eq:an}
                \lim_{n\to\infty} \frac{\ln n}{a_n} \ = \ \lim_{n\to\infty} \frac{a_n}{n} \ = \ 0 \ .
            \end{equation}
        \item If $lim_{n\to\infty} \rho_n =0$, then $\avg[M_n]$ vanishes faster than any polynomial, i.e.,	
            \begin{equation}\label{eq:secondMnbound}
                \avg[M_n] \ = \ O \left ( \frac{1}{n^k} \right ) \ .
            \end{equation}
            for each $k \in \mathbb{N}$,
    \end{enumerate}
\end{lemma}

Lemma \ref{boundsForSubMultiplicativeAverages}  directly implies bounds on $\avg[c(\Phi^{(n)})]$, stated in the following
\begin{lemma}\label{mixingcondtionscor}
    Suppose that Assumptions \ref{1} and \ref{2} hold, and let $\alpha_n$ and $\rho_n$ be mixing coefficients defined as defined in eqs.\ \eqref{alphan} and     \eqref{rhon}, respectively. For $r \in (0,1)$ define
    \begin{equation*}
        \tau_r \ = \ \inf\{n\geq 1: \cnum{\Psi^{(n)}}\le r \ \& \ \cnum{\Phi^{(n)}} \leq r\} \ .
    \end{equation*}
    Then we have that $\tau_r < \infty$ almost surely. Moreover
    \begin{enumerate}
        \item If $\sum_{k=1}^\infty \alpha_k^{1/\lambda} < \infty$, for some $\lambda > 0$, then 
            \begin{equation}\label{upperboundTauronAlpha}
                    \max\left \{\prob[\tau_r > n]\ , \ \avg[\cnum{\Phi^{(n)}}] \right \} \ = \ O\bigg( \frac{a_n}{n}\bigg)^\lambda \ 
            \end{equation}
            for any sequence $(a_n)_{n\ge 1}$ satisfying eq.\ \eqref{eq:an}.
        \item If $\lim_{n\to\infty} \rho_n =0$, then 
            \begin{equation}\label{upperboundTauonRho}
                \max\left \{\prob[\tau_r > n]\ , \ \avg[\cnum{\Phi^{(n)}}] \right \}  \ = \ O\left(\frac{1}{n^k}\right ) \ 
            \end{equation}	
            for any $k\ge 1$.
    \end{enumerate}
\end{lemma}
\begin{proof}
    From Corollary \ref{cConvergeTo0} and Lemma \ref{contractionCoefficientPsi} we see that $\prob[\tau_r < \infty] = 1$. We also have, by Assumption \ref{2}, that $\phi_n$ is non-destructive and non-transient for all $n\geq 1$, with probability one. Therefore, we have that almost surely for all $n\in\mathbb{Z}$, $\cnum{\Psi^{\tau_r+n}} , \cnum{\Phi^{\tau_r+n}} < r$. 

    Suppose that  $\sum_{k=1}^\infty \alpha_k^{1/\lambda} < \infty$.  We start with the observation that $\alpha_n$ is non-increasing in $n$; this can be seen directly from the definition \eqref{alphan} of $\alpha_n$ using the fact that $(\mathcal{F}^n)_{n\geq1}$ is decreasing in $n$. Since $\alpha_n^{1/\lambda}$ is a non-increasing sequence of positive numbers with $\sum_{n}\alpha_n^{1/\lambda} <\infty$, we have $\lim_{n\to\infty} n\alpha_n^\lambda = 0$. Therefore we have   $\alpha_n \le cn^{-\lambda}$. Now notice that the sequence $M_n= c(\Phi^{(n)})$, for $n\ge 1$, satisfies the sub-multiplicative condition in Lemma \ref{boundsForSubMultiplicativeAverages}. Therefore we obtain
    \begin{equation}
        \avg[\cnum{\Phi^{(n)}}] \ \leq\ c_1 \ \bigg( \frac{a_n }{n}\bigg)^\lambda \ 
    \end{equation}
    for any sequence $(a_n)_{n\ge1}$ satisfying eq.\ \eqref{eq:an}.
    A similar analysis can be applied to $(c(\Psi^{n})_{n\in\mathbb{N}}$, resulting in
    \begin{equation}
        \avg[\cnum{\Phi^{(n)}}] \ \leq \ c_2 \bigg( \frac{a_n}{n}\bigg)^\lambda \ .
    \end{equation}
    Since 
    \begin{equation}\label{tauRbound}
        \prob[\tau_r > n] \ \le \ \prob[c(\Psi^{(n)}) > r] + \prob[c(\Phi^{(n)}) > r] \ \leq \ \frac{1}{r} \avg[c(\Phi^{(n)})] +\frac{1}{r} \avg[c(\Psi^{(n)})] \ ,
    \end{equation}
    we see that eq.\ \eqref{upperboundTauronAlpha} holds.
	
    If $\lim_{n\to\infty} \rho_n = 0$, then the second part of Lemma \ref{boundsForSubMultiplicativeAverages} applies and eq.\ \eqref{tauRbound} still holds. These two combined give us eq.\ \eqref{upperboundTauonRho}. 
\end{proof}

We are now ready to state the main technical estimate of this section:
\begin{lemma}\label{mixingconditionslemma}
     Suppose that Assumptions \ref{1} and \ref{2} hold, and let $\alpha_n$ and $\rho_n$ be mixing coefficients defined in eqs.\ \eqref{alphan} and     \eqref{rhon}, respectively.  Let $r \in (0,1)$ and let $\tau_r$ be as defined in Lemma \ref{mixingcondtionscor}. Let $n_\alpha$ denote the integer part of $(1-\alpha)n$, for $\alpha \in (0,1)$.
    \begin{enumerate}
        \item If Assumption \ref{2'} holds with $p>2$ then there is $K<\infty$ such that
            \begin{equation}\label{eq:Lqbound}
                \norm{\cond{\xi_0}{n}}_{q} \ \leq \ K \left [ \alpha_{n-n_\alpha}^{(p-2)/p} + \avg[\cnum{\Phi^{(n_\alpha)*}}] + (\prob[\tau_r > n_\alpha])^{1/q} \right ]\ ,
            \end{equation}
        with $q$ the conjugate exponent to $p$.
        \item If Assumption \ref{2'} holds with $p=2$ then there is $K<\infty$ such that
            \begin{equation}\label{eq:L2bound}
                \norm{\cond{\xi_0}{n}}_2 \ \leq  \ K \left [ \rho_{n-n_\alpha} + \avg[\cnum{\Phi^{(n_\alpha)*}}] + (\prob({\tau_r > n_\alpha }))^{1/2} \right ] \ .
            \end{equation}
    \end{enumerate}
\end{lemma}

Before proving Lemma \ref{mixingconditionslemma}, let us show how it implies Theorem \ref{customthm3}.
First note that if $(b_n)_{n\ge 1}$ is a sequence of non-negative numbers then
\begin{align}
\label{eq:partialsum} &\sum_{n=1}^\infty b_{n_\alpha} \ \le \ \frac{1}{1-\alpha} \sum_{n=1}^\infty b_n \ , \\
\label{eq:partialsum2}
&\sum_{n=1}^\infty b_{n-n_\alpha} \ \le \ \frac{1}{\alpha} \sum_{n=1}^\infty b_n \ .
\end{align}
To see that eq.\ \eqref{eq:partialsum} holds, note that given $m\in \mathbb{N}$, the number of integers $n$ such that $n_\alpha=m$ is bounded by $\frac{1}{1-\alpha}$.  The proof of eq.\ \eqref{eq:partialsum2} is similar.
Now suppose that Assumption \ref{2'} holds with $p>2$ and $\sum_{n}\alpha_n^{(p-2)/p} < \infty$.  Then by Lemma \ref{mixingcondtionscor}, Lemma \ref{mixingconditionslemma}, and eqs.\ (\ref{eq:partialsum}, \ref{eq:partialsum2}), we have
$$ \sum_{n=1}^\infty \norm{\cond{\xi_0}{n}}_{q} \ \le \  K'  \sum_{n\ge 1} \left [ \alpha_n^{\frac{p-2}{p}} + \left ( \frac{a_n}{n} \right )^{\frac{p}{p-2}}  + \left ( \frac{a_n}{n} \right )^{\frac{p}{p-2} \frac{1}{q}}\right ]$$
for a suitable $K'<\infty$ and a slowly increasing sequence $(a_n)_{n\ge 1}$ satisfying eq.\ \eqref{eq:an}.  Since $\frac{p}{p-2}>1$ and $\frac{p}{p-2}\frac{1}{q}=\frac{p-1}{p-2}>1$ the right hand side is finite. Similarly, if $\sum_{n} \rho_n < \infty$, then we have
$$ \sum_{n=1}^\infty \norm{\cond{\xi_0}{n}}_{q} \ \le \  K'  \sum_{n\ge 1} \left [ \rho_n + \frac{1}{n^k} \right ] $$
for any $k$, which is clearly finite.  This completes the proof of Theorem \ref{customthm3}.

We now turn to the proof of Lemma \ref{mixingconditionslemma}:

\begin{proof}[Proof of Lemma \ref{mixingconditionslemma}]
    By Lemma \ref{existanceofz}, we have $\Phi^{(n_{\alpha})} \cdot Z_{n_{\alpha}+1} = Z_1$. Therefore
    $$ \xi_0 \ = \ A_n + B_n - \avg[A_n] \ ,$$
    where
    \[A_n = \ln \norm{\varphi_0^*(\Phi^{(n_{\alpha})} \cdot Z_{n_{\alpha}+1})} - \ln \norm{\varphi_0^*(\Phi^{(n_{\alpha})} \cdot \tfrac{1}{D}\mathbb{I}}\]
    and
    $$B_n = \ln \norm{\varphi_0^*(\Phi^{(n_{\alpha})} \cdot \tfrac{1}{D} \mathbb{I})} - \avg\left [\ln \norm{\varphi_0^*(\Phi^{(n_{\alpha})} \cdot \tfrac{1}{D}\mathbb{I})} \right ] \ .$$
					
     Now consider the event $\{\tau_r \leq n_\alpha\}$. On this event, $\cnum{\Phi^{(n_\alpha)*}} \leq r$. To bound $A_n$ on this event we will use the following proposition which we prove below after completing the present proof.
     \begin{prop}\label{upperboundForLnDifferenceAtStoppingTimeTauR}
    Let $\psi, \phi \in \mapspace$. Suppose that $\psi$ is a positive map and $\phi$ is a strictly positive map with  $\cnum{\phi} \leq r < 1$. If $\psi$ is non-transient, then for any $A,B \in \mathbb{S}_D$ we have 
    \begin{equation*}
        \big| \ln \norm{\psi\big(\phi\cdot A\big)} - \ln \norm{\psi\big(\phi\cdot B\big)}  \big| \ \leq \ c(\phi)	 \frac{2}{r}\ln\frac{1}{1-r} 			
    \end{equation*}
    \end{prop}
    \noindent Using Proposition \ref{upperboundForLnDifferenceAtStoppingTimeTauR}, we see that
    $$|A_n| \ \leq \ \frac{2}{r} \ln \frac{1}{1-r} \cnum{\Phi^{(n_\alpha)*}} 1_{\tau_r \le n_\alpha} + 2\left (|\ln\norm{\varphi_0^*}| + |\ln v(\varphi^*_0)| \right ) 1_{\tau_r > n_\alpha} \ =: \ A_n' .$$
    Therefore, using H\"older's inequality and Assumption \ref{2'}, we have
    \begin{equation}\label{eq:An'}
        \avg|A_n| \ \leq \ \avg A_n' \ \leq \ C \left ( \avg[\cnum{\Phi^{(n_\alpha)*}}] + (\prob({\tau_r > n_\alpha }))^{(p-1)/p} \right  )
    \end{equation}
    with $C <\infty$. Furthermore we also have that 
    \begin{equation}\label{eq:AnBnLp}
        \sup_n \norm{A_n}_p \ \le \ \sup_n \norm{A_n'}_p \ < \ \infty \text{ and }  \sup_n\norm{B_n}_p \ < \ \infty \ . 
    \end{equation}
					
    Now, for $\frac{1}{p}+ \frac{1}{q}=1$ we have that 
    \begin{equation}\label{eq:variationLq}
        \norm{\cond{\xi_0}{n}}_q  =  \sup_{\{f\in L^p(\mathcal{F}^n):\norm{f}_p \le 1\}} \big| \avg[f\xi_0]\big| \ . 
    \end{equation}
    Hence to bound $\norm{\cond{\xi_0}{n}}_q  < \infty$ it suffices to find a uniform upper bound for $\avg[\xi f]$ as $f$ ranges over the unit ball in $L^p(\mathcal{F}^n) .$ Since $\xi_0 = A_n + B_n - \avg[A_n]$ and $\avg[B_n]=0$, we have
    \begin{align}
        |\avg[\xi_0f]|\	&\leq \ \big |\avg  [A_nf] \big | + \big|\avg[B_nf]\big| + \big|\avg[A_n]\avg[f]\big| \nonumber\\
								&\leq \ \avg[A_n'\,|f|] + \big|\avg[B_nf]-\avg[B_n]\avg[f] \big| + \avg[A_n']\avg[|f|] \nonumber\\
								&\leq \  \big|\avg[A_n'\, |f|] - \avg[A_n']\avg[|f|] \big| + \big|\avg[B_nf] - \avg[B_n]\avg[f]|\big| + 2\avg[A_n']\avg[|f|] \label{eq:xi0f}
    \end{align}

    To estimating the right hand side we use the following covariance inequalities involving the mixing coefficients $\alpha_n$ and $\rho_n$.
    \begin{lemma}[{\cite[{\S1.2 Theorem 3}]{doukhan2012mixing} \textemdash see also \cite[\S6.2]{hennion1997limit}}]\label{lem:rhoandalpha}
        For each $n\in \mathbb{N}$, Let $\alpha_n$ and $\rho_n$ be as defined in eqs.\ \eqref{alphan} and  \eqref{rhon}, respectively.  For each $n,k\in \mathbb{N}$, we have 
        \begin{equation*}
            \big| \avg[XY] - \avg[X]\avg[Y] \big| \ \leq \ \rho_n \norm{X}_2\norm{Y}_2
        \end{equation*}
        whenever $X \in L^2(\mathcal{F}_k)$ and $Y \in L^2(\mathcal{F}^{n+k})$, and
        \begin{equation*}
            \big| \avg[XY] - \avg[X]\avg[Y] \big| \ \leq \ 8\alpha_n^{1/r}\norm{X}_p\norm{Y}_q
        \end{equation*}
        whenever $X \in L^p(\mathcal{F}_k)$ and $Y \in L^q(\mathcal{F}^{k+n})$ with $p,q,r \in [1,\infty]$ such that $\frac{1}{p} + \frac{1}{q}+\frac{1}{r} = 1$.
    \end{lemma}

    We note that $A_n,B_n\in L^p(\mathcal{F}_{n_\alpha})$.  If $p >2$, then eq.\ \eqref{eq:xi0f} and Lemma \ref{lem:rhoandalpha} (with $q=p$ and $r=\frac{p}{p-2}$) together imply that
    $$|\avg[\xi_0f]| \ \leq \ 8\alpha_{n-n_\alpha}^{\frac{p}{p-2}}\left (\norm{A'_n}_{p} + \norm{B_n}_{p} \right ) \norm{f}_{p} + 2\avg[A_n']\norm{f}_{p} \ ,$$
    where we have used the estimate $\avg[|f|]\le \norm{f}_p$ in the last term. Eq.\ \eqref{eq:Lqbound} follows this inequality together with eqs.\ (\ref{eq:An'}, \ref{eq:AnBnLp}, \ref{eq:variationLq}). If $p=2$, then eq.\ \eqref{eq:xi0f} and Lemma \ref{lem:rhoandalpha} together imply that
    $$|\avg[\xi_0f]| \ \leq \ \rho_{n-n_\alpha} \left (\norm{A'_n}_{2} + \norm{B_n}_{2} \right ) \norm{f}_{2} + 2\avg[A_n']\norm{f}_{2} \ .$$
    Eq.\ \eqref{eq:L2bound} follows from this inequality, again using eqs.\ (\ref{eq:An'}, \ref{eq:AnBnLp}, \ref{eq:variationLq}). This completes the proof of Lemma \ref{mixingconditionslemma}.
\end{proof}

It remains to prove Proposition \ref{upperboundForLnDifferenceAtStoppingTimeTauR}:
   \begin{proof}[Proof of Proposition \ref{upperboundForLnDifferenceAtStoppingTimeTauR}]
    From \cite[Lemma 3.3]{movassagh2022ergodic}, the quantity $m(A,B)$ appearing in the definition \eqref{distance} of the metric $\dis{A}{B}$ can be expressed as
    $$ m(A,B) \ = \ \min \left \{  \frac{\tr [XA]}{\tr[XB]} \ : \ X\in \mathbb{S}_D \text{ and } \tr[XA] \neq 0 \right \} \ . $$
    Since
    $$ \frac{\norm{\psi(\phi\cdot A)}}{\norm{\psi(\phi\cdot B)}} \ = \ \frac{\tr \psi^*(\mathbb{I}) \phi \cdot A}{\tr \psi^*(\mathbb{I}) \phi \cdot B} \ = \ 
    \frac{\tr \psi^*(\frac{1}{D}\mathbb{I}) \phi \cdot A}{\tr \psi^*(\frac{1}{D}\mathbb{I}) \phi \cdot B} \ ,$$
    we see that 
    $$ m(\phi\cdot A, \phi\cdot B) \ \le \ \frac{\norm{\psi(\phi\cdot A)}}{\norm{\psi(\phi\cdot B)}} \ \leq \ \frac{1}{m(\phi\cdot B,\phi\cdot A)} \ .$$
    Since $\phi\cdot A$, $\phi\cdot B$ are positive definite (because $\phi$ is strictly positive), the various terms appearing in this inequality are all finite and non-zero. Taking logarithms yields 
    \begin{multline*}
        \big| \ln \norm{\psi\big(\phi\cdot A\big)} - \ln \norm{\psi\big(\phi\cdot B\big)}  \big| \
            \leq \  -\ln m(\phi\cdot A, \phi\cdot B) -\ln m(\phi\cdot B, \phi\cdot A)\\
        \leq \  \ln \frac{1+\dis{\phi\cdot A}{\phi\cdot B}}{1-\dis{\phi\cdot A}{\phi\cdot B}} \ \le \
             \ln \frac{1+\cnum{\phi}}{1-\cnum{\phi}} \ ,
    \end{multline*}
    where we have used the definition eq.\ \eqref{distance} of $\dis{\cdot}{\cdot}$ and Lemma \ref{equvilanceOfMatrics} to obtain $\dis{\phi\cdot A}{\phi\cdot B} \leq c(\phi)$.
    Now for $x \in [0,1)$ we have 
    $$ \frac{1+x}{1-x} \ \le \  \frac{1}{(1-x)^2}$$
    As $x = \cnum{\phi} \in (0,1)$ (since $\phi$ is strictly positive) we have that
    \begin{equation*}
        \big| \ln \norm{\psi\big(\phi\cdot A\big)} - \ln \norm{\psi\big(\phi\cdot B\big)}  \big| \ \leq \ 2\ln \frac{1}{1-\cnum{\phi}} \ . 
    \end{equation*}
    
    Now consider the convex function $f(x) = \ln 1/(1-x)$ for $x \in [0,1)$. Since $f$ is convex and $f(0)=0$, we have $f(tr) \le tf(r)$ for any $t,r\in[0,1)$. Hence, $f(\lambda) \leq f(r) \lambda/r$ for any $\lambda \in [0,r]$. Thus
    \begin{equation*}
        \big| \ln \norm{\psi\big(\phi\cdot A\big)} - \ln \norm{\psi\big(\phi\cdot B\big)}  \big| \ \leq \ c(\phi)	 \frac{2}{r}\ln\frac{1}{1-r} 			
    \end{equation*}  
    if $c(\phi)\le r$.
\end{proof}

\addcontentsline{toc}{section}{Bibliography}
\bibliographystyle{amsplain}
\bibliography{ref}


\end{document}